\documentclass[letterpaper, 10 pt, conference]{IEEEtran}  
\IEEEoverridecommandlockouts                              
\usepackage{graphicx} 
\usepackage{times} 
\usepackage{amsmath} 
\usepackage{amssymb}  
\usepackage{mathtools}
\usepackage{epstopdf}
\usepackage{subfig}
\usepackage{cite}
\usepackage[hidelinks]{hyperref}
\usepackage{framed}
\usepackage{multirow}
\usepackage{xcolor}
\usepackage{multirow}
\usepackage{etoolbox}
\usepackage{mathptmx}
\usepackage{calrsfs}
\usepackage{lipsum}
\usepackage{amsthm}
\newtheorem{theorem}{Theorem}

\newtheorem{corollary}{Corollary}

\newtheorem{remark}{Remark}

\newtheorem{problem}[theorem]{Problem}

\newcommand{\trace}{\mbox{\rm trace}}

\newcommand{\zpos}{\mbox{$\mathbb{Z}_{+}$}}

\newcommand{\rn}{\mbox{$\mathbb{R}^n$}} 
\newcommand{\rnn}{\mbox{$\mathbb{R}^{n \times n}$}}
\newcommand{\rank}{\mbox{$\rm rank$}}

\newcommand{\rar}{\rightarrow}

\newcommand{\tri}{\triangleq}

\newcommand{\be}{\begin{equation}}
\newcommand{\ee}{\end{equation}}
\newcommand{\bea}{\begin{eqnarray}}
\newcommand{\eea}{\end{eqnarray}}
\newcommand{\bes}{\begin{eqnarray*}}
\newcommand{\ees}{\end{eqnarray*}}
\newcommand{\bi}{\begin{itemize}}
\newcommand{\ei}{\end{itemize}}
\newcommand{\ben}{\begin{enumerate}}
\newcommand{\een}{\end{enumerate}}
\newcommand{\bp}{\begin{problem}}
\newcommand{\ep}{\end{problem}}
\newcommand{\hso}{\hspace{.1in}}
\newcommand{\hst}{\hspace{.2in}}

\title{\LARGE \bf  Characterization of the Gray-Wyner Rate Region for Multivariate Gaussian Sources: Optimality of  Gaussian Auxiliary RV }

\author{Evagoras Stylianou, Charalambos D. Charalambous and Jan H. van Schuppen 
\thanks{E. Stylianou is with  Technical University of Munich. C. D. Charalambous is with  University of Cyprus. J. H. van Schuppen is with Van Schuppen Control Research, Gouden Leeuw 143, 1103 KB Amsterdam. E-mails:
   chadcha@ucy.ac.cy,jan.h.van.schuppen@xs4all.nl,evagoras.stylianou@tum..de}
}
 \IEEEaftertitletext{\vspace{-3.1ex}}
\begin{document}
\maketitle

\begin{abstract}
Examined in this paper,  is the Gray and Wyner achievable lossy rate region 
 for a tuple of correlated multivariate    Gaussian random variables (RVs) $X_1 : \Omega \rightarrow {\mathbb R}^{p_1}$ and $X_2 : \Omega \rightarrow {\mathbb R}^{p_2}$
 with respect to square-error distortions at the two decoders.  It is shown that 
 among all joint distributions induced by  a  triple of RVs $(X_1,X_2, W)$, such that  $W : \Omega \rightarrow {\mathbb W} $ is the auxiliary RV taking  continuous, countable, or finite values, the   Gray and Wyner achievable rate region is characterized by jointly Gaussian RVs $(X_1,X_2, W)$ such that $W $ is an $n$-dimensional Gaussian RV.  It then follows that the achievable rate region is  parametrized  by the three  conditional covariances $Q_{X_1,X_2|W}, Q_{X_1|W}, Q_{X_2|W}$ of the jointly Gaussian RVs. Furthermore, if the RV  $W$ makes $X_1$ and $X_2$ conditionally independent, then the corresponding subset of the achievable rate region, is simpler, and  parametrized  by only the two  conditional covariances $Q_{X_1|W}, Q_{X_2|W}$. The paper also includes the characterization of the Pangloss plane of the Gray-Wyner rate region along with the characterizations of the corresponding rate distortion functions, their test-channel distributions, and  structural properties of the realizations which induce these distributions.
\end{abstract}

\section{Introduction,  Literature Review, Main Results}
\label{sec.Intro}
  Gray and Wyner in \cite{gray-wyner:1974} considered the {\it the simple network} of Fig. \ref{fig:gray}, and characterized the achievable (lossy)  rate region,     for  an arbitrary tuple of sources, modeled by jointly independent  random variables (RVs)  $(X_1^N, X_2^N)= \big\{(X_{1,t}, X_{2,t}): t=1,2, \ldots,N\big\}$,  with two distortions functions at the decoders. 
They characterized the   operational rate region, denoted by  ${\cal R}(\Delta_1, \Delta_2)$,    by a coding scheme that uses an  auxiliary RV,  $W: \Omega \rar {\mathbb W}$, where  ${\mathbb W}$ is an arbitrary space, via   the family of probability distributions ${\cal P}$ induced by $(X_1, X_2, W)$ on their corresponding measurable spaces ${\mathbb X}_1 \times {\mathbb X}_2\times {\mathbb W}$ defined by, 
\begin{align*}
{\cal P} \tri
    \Big\{ 
          {\bf P}_{X_1, X_2, W}
          \Big|  \; \mbox{the $(X_1,X_2)$-marginal distr. is  ${\bf P}_{X_1, X_2}$}.
      \Big\}
\end{align*}
Specifically, in \cite[Theorem 8]{gray-wyner:1974}, they defined, for each ${\bf P}_{X_1, X_2, W} \in {\cal P}$ and  distortions $\Delta_1 \geq 0, \Delta_2 \geq 0$, the subset of the  Euclidean $3$-dimensional  space,  
\begin{align}
  {\cal R}^{{\bf P}_{X_1,X_2,W}}
   (\Delta_1, \Delta_2) = 
    &\Big\{
      \big(
        R_0,R_1,R_2
      \big)
         \in \mathbb{R}_+^3 
      \Big| \ \  R_0 \geq I(X_1, X_2; W),
     \nonumber \\ 
  & 
      R_1 \geq R_{X_1|W}(\Delta_1), \ \ R_2 \geq R_{X_2|W}(\Delta_2) 
    \Big\}, \label{eq_32} 
\end{align} 
where  $\mathbb{R}_+\tri [0,\infty)$, and    $R_{X_i|W}(\Delta_i)$ 
 denotes the conditional
rate distortion function (RDF) of $X_i$, conditioned on $W$, at decoder $i$, for   $i=1,2$. The Gray-Wyner achievable operational lossy rate region  ${\cal R}(\Delta_1, \Delta_2)$  is then characterized by,  ${\cal R}(\Delta_1, \Delta_2) ={\cal R}^{*}(\Delta_1, \Delta_2)$, where,
\begin{figure}[t]
    \centering
    \includegraphics[width = \columnwidth ,height = 3cm]{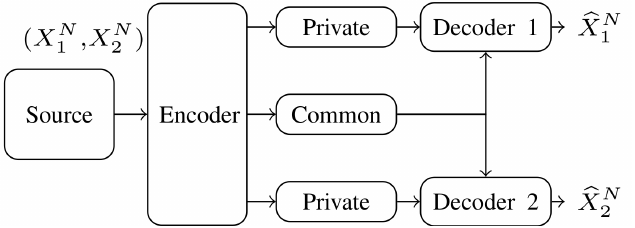}
      \caption{ The Gray-Wyner network \cite{gray-wyner:1974} with Private Rates $R_1$ and $R_2$ and Common Rate $R_0$. }
\label{fig:gray}\vspace{-0.4cm}
\end{figure}
\begin{align}
{\cal R}^{*}(\Delta_1, \Delta_2) \tri \Big(\bigcup_{ {\bf P}_{X_1,X_2, W} \in {\cal P}} {\cal R}^{{\bf P}_{X_1,X_2,W}}(\Delta_1, \Delta_2)\Big)^\mathrm{c}, \label{eq_32_a} 
\end{align}
and where $(\cdot)^\mathrm{c}$ denotes the closure of the  indicated set. \\
Gray and Wyner also  proved that,     ${\cal R}(\Delta_1, \Delta_2)$, can be alternatively determined  from  \cite[(4) of page 1703,  eqn(42)]{gray-wyner:1974},\vspace{-0.1cm}
\begin{align}
T(\alpha_1, \alpha_2) = \inf_{{\bf P}_{X_1,X_2, W} \in {\cal P}} \Big\{&I(X_1, X_2; W)+\alpha_1 R_{X_1|W}(\Delta_1) \nonumber \\
& + \alpha_2 R_{X_2|W}(\Delta_2)\Big\}, \label{eqn_new_3}
\end{align}
where $0\leq \alpha_i\leq 1,\; i=1,2$ and $\alpha_1+\alpha_2\geq 1$.\\
Moreover, it was shown in \cite[Theorem~6]{gray-wyner:1974}, that if a triple $(R_0, R_1, R_2)$ lies on the Gray-Wyner lossy rate region, i.e.,  $(R_0, R_1, R_2) \in {\cal R}(\Delta_1, \Delta_2)$, then it satisfies the following bounds \vspace{-0.5cm}
\begin{align}
& R_0 + R_1 +R_2 \geq R_{X_1, X_2}(\Delta_1, \Delta_2),   \label{eq_32a}  \\
& R_0 +R_1  \geq R_{X_1}(\Delta_1), \hso  R_0 +R_2  \geq R_{X_2}(\Delta_2), \label{eq_32c} 
\end{align}
where  $R_{X_1,X_2}(\Delta_1,\Delta_2)$ is the joint RDF of $(X_1,X_2)$ with  joint decoding of $(X_1, X_2)$ by $(\widehat{X}_1, \widehat{X}_2)$, and $R_{X_i}(\Delta_i)$ is the marginal RDF of $X_i$ at decoder $i$, for $i=1,2$. 
The set of rate triples $(R_0, R_1, R_2) \in {\cal R}(\Delta_1, \Delta_2)$ 
 which satisfy  $\sum_{i=1}^3 R_i= R_{X_1, X_2}(\Delta_1, \Delta_2)$ 
 is called the  {\em Pangloss Plane}. A rate-triple $(R_0,R_1,R_2) \in  {\cal R}(\Delta_1,\Delta_2)$ that lies on the Pangloss plane was computed by the authors in \cite[Section~2.5, (B)]{gray-wyner:1974} making use Gray's \cite{gray:1973,gray-wyner:1974} compound joint RDF of $(X_1, X_2)$, for a tuple of scalar-valued  Gaussian RVs with square-error distortions. 

More recently, progress is reported by Viswanatha, Akyol and Rose \cite{viswanatha:akyol:rose:2014}, and Xu, Liu, and Chen \cite{xu:liu:chen:2016:ieeetit}, on the  characterization of the lossy common information which is defined as the minimum common message rate $R_0$   on  the Gray-Wyner lossy rate region, when the sum rate is arbitrary close to the joint RDF.  The lossy common information is computed  for a tuple of scalar-valued  Gaussian RVs, ${\mathbb X}_i={\mathbb R}, i=1,2$, with square-error distortions, by  making use of  Xiao's and Luo's \cite{xiao-luo:2005}  closed-form expression of the joint RDF of a tuple of scalar-valued Gaussian RVs.  Other related investigations are, Wyner's common information \cite{wyner:1975} which is the lossless counterpart of the lossy common information, and related papers by  Witsenhausen  \cite{witsenhausen:1976b,witsenhausen:1975:pairsrv}, G\'{a}cs and K\"{o}rner \cite{gacs-korner:1973}, Satpathy and Cuff  \cite{satpathy:cuff:2012}, Veld and Gastpar \cite{veld-gastpar:2016} and Sula and Gastpar 
\cite{sula:gastpar:2019:arxiv}. Previous work of the authors on the   Gray-Wyner rate region is found in \cite{charalambous:schuppen:2019:arxiv},  \cite{charalambous2020characterization}. 

However, the fundamental problem of characterizing ${\cal R}(\Delta_1, \Delta_2) $ for multivariate sources, remains to this date an open problem.

{\it Main Results of this Paper. }
 We characterize ${\cal R}(\Delta_1, \Delta_2) $, and  provide an answer to a long standing open problem,   for sources modeled by a  tuple of multivariate  jointly independent and identically distributed Gaussian RVs, with respect to square-error distortions,  i.e., \vspace{-0.1cm} 
\begin{align}
&X_{i} : \Omega \rightarrow {\mathbb R}^{p_i}= {\mathbb X}_i, \;   D_{X_i} (x_i, \widehat{x}_i)\tri    ||x_{i}-\widehat{x}_{i}||_{{\mathbb R}^{p_i}}^2, \; i=1,2
 \label{dist_1} \\
&{\bf P}_{X_{1}, X_{2}} \hso \mbox{is a jointly Gaussian distribution}\label{dist_3}
\end{align} 
where $p_1,p_2$ are finite positive integer numbers and $||\cdot||_{{\mathbb R}^{p_i}}^2$ are Euclidean distances on ${\mathbb R}^{p_i}, i=1,2$.

The paper includes  the following main results.\\
(1) Theorem~\ref{thm_par_g} and Theorem~\ref{them_g}   which state that: \\
(i) The Gray-Wyner lossy rate region ${\cal R}(\Delta_1, \Delta_2)$ and $T(\alpha_1,\alpha_2)$ are characterized by replacing the set ${\cal P}$ in \eqref{eq_32_a} and \eqref{eqn_new_3}  with the subset ${\cal P}^\mathrm{G}$ of Gaussian distributions defined by, \vspace{-0.1cm}
 \begin{align}
{\cal P}^\mathrm{G} \triangleq & \Big\{ {\bf P}_{X_1, X_2, W}\Big|  X_i : \Omega \rightarrow  {\mathbb R}^{p_i}, i=1,2,  \: W : \Omega \rightarrow  {\mathbb R}^n, \nonumber \\
& \mbox{ $(X_1,X_2) \in G(0, Q_{(X_1,X_2)})$ is fixed}, \: W \in G(0, Q_W),   \nonumber \\
&  {\bf P}_{X_1, X_2, W}\in  G(0, Q_{(X_1,X_2,W)}) 
\Big\} \subset {\cal P},    \label{eqn_new_5}\\
 & Q_{(X_1,X_2,W)} =
  \left(
  \begin{array}{lll}
    Q_{X_1}   & Q_{X_1,X_2}   & Q_{X_1,W} \\
    Q_{X_1,X_2}^T & Q_{X_2}   & Q_{X_2,W} \\
    Q_{X_1,W}^T   & Q_{X_2,W}^T   & Q_W
  \end{array}
  \right).   \label{jdf_g}
\end{align}
The notation, ${\bf  P}_{X_1,X_2, W} \in G(0, Q_{(X_1,X_2,W)})$ means the joint distribution is the Gaussian distribution of the vector $(X_1^{T} \; X_2^T\; W^T)^T$, with zero mean and covariance matrix $Q_{(X_1,X_2,W)}$.  \\
(ii) The Gray-Wyner lossy rate region  ${\cal R}(\Delta_1, \Delta_2)$ and $T(\alpha_1,\alpha_2)$ are parametrized  by the conditional covariances $Q_{X_1,X_2|W}, Q_{X_i|W}, i=1,2$, where $W: \Omega \rar {\mathbb R}^n$ is a Gaussian RV. \\
 (2) Corollary~\ref{cor_sgy}, which states that, if   $W$ makes  $X_1$ and $X_2$ conditional independent, then the  subset of ${\cal R}(\Delta_1, \Delta_2)$,  is    parametrized   only with respect to $ Q_{X_i|W}, i=1,2$.\\
  (3) Theorem~\ref{wyner_common_g} and Theorem~\ref{th:crdf_g} that identify achievable lower bounds on  the mutual information $I(X_1,X_2;W)$ and RDFs $R_{X_i|W}(\Delta_i), i=1,2,\;R_{X_1,X_2|W}(\Delta_1, \Delta_2)$. Specifically, it is shown that,  among all triples of RVs $(X_1,X_2, W)$, with arbitrary $W : \Omega \rightarrow {\mathbb W} $,  these lower bounds are achieved  if 
 $(X_1,X_2, W)$ are jointly Gaussian, i.e., $W : \Omega \rightarrow {\mathbb R}^n $ is a Gaussian RV. Moreover, realizations of $(X_1,X_2,\widehat{X}_1, \widehat{X}_2,W)$ that achieve the aforementioned bounds are also provided. \\
(4) Theorem \ref{thm_pangloss} which characterize the \emph{Pangloss plane} of the Gray-Wyner network by identifying conditions on the joint distribution $\mathbf{P}_{X_1,X_2,\widehat{X}_1, \widehat{X}_2,W}$ such that $R_0+R_1+R_2 = R_{X_1,X_2}(\Delta_1,\Delta_2)$.

 However, further
research is required to carry out the remaining optimizations and computations, which are involved in  the characterizations of the ${\cal R}(\Delta_1, \Delta_2)$ and $T(\alpha_1,\alpha_2)$. These calculations are expected to be challenging, because they require  closed-form expressions of the RDFs, $R_{X_i|W}, i=1, 2, R_{X_1,X_2}(\Delta_1,\Delta_1)$, and the structural properties of their test-channel realizations   \cite{evagoras-cdc-tc:2021,gkangos-charalambous:2021,gkagkos2020structural}.

\section{Parametrization of Gray and Wyner Rate Region of Gaussian RVs}
\label{sect:wcilossy}
{\it Notation.} { $\mathbb{Z}_+ = \{ 1,2, \ldots, \}$,  
 $\mathbb{N} = \{ 0,1,2, \ldots, \}$.}
Denote the real numbers by
$\mathbb{R}$. For $x \in {\mathbb R},  (x)^+\tri \max\{1,x\}$
The vector space of $n$-tuples of real numbers is denoted 
by $\rn$.
Denote the Borel $\sigma$-algebra on this vector space
by $B(\mathbb{R}^n)$
hence $(\mathbb{R}^n,B(\mathbb{R}^n))$
is a measurable space. The expression $\mathbb{R}^{n \times m}$
denotes the set of $n$ by $m$ matrices with elements
in the real numbers, for $n, ~ m \in \zpos$.  An {\em $\rn$-valued Gaussian} RV, is denoted by $X \in G(m_X, Q_X)$, where $m_X \in \rn$ is the
 {\em mean value},   and 
$Q_X \in \rnn$, $Q_X = Q_X^T \succeq 0$  the {\em variance}.
The {\it effective dimension} of the RV is denoted by 
$\dim (X) = \rank (Q_X)$. An $n\times n$ identity matrix is denoted by $I_n$. 
For a tuple of Gaussian RVs $(X_1,X_2) \in G(0,Q_{(X_1,X_2)})$, its variance matrix $Q_{(X_1,X_2)}$  is defined as in    (\ref{jdf_g}), with $W$ removed. 
The variance $Q_{(X_1,X_2)}$ is distinguished 
from $Q_{X_1,X_2}$.

In order to prove our main results, i.e., the characterizations of ${\cal R}(\Delta_1, \Delta_2)$ and $T(\alpha_1,\alpha_2)$ in Theorem \ref{thm_par_g} and Theorem \ref{them_g},  we need two intermediate results. Specifically, lower bounds on the mutual information  $I(X_1,X_2;W)$ and the conditional RDFs $R_{X_i|W}(\Delta_i),
\;i=1,2$, and realizations that achieve these bounds. We begin with the achievable lower bounds on $I(X_1,X_2;W)$.

 \begin{theorem}
 \label{wyner_common_g}
Consider a tuple of Gaussian RVs
$X_i: \Omega \rightarrow \mathbb{R}^{p_i}$, 
with  $(X_1 , X_2 ) \in G(0,Q_{(X_1,X_2)})$ such that $Q_{(X_1,X_2)}\succ 0$ (which implies $Q_{X_i} \succ 0$ for $ i=1, 2$).
 Let  $W: \Omega \rightarrow {\mathbb W}$ be any {\em auxiliary RV},  with $({\mathbb W}, B({\mathbb W}))$ an arbitrary measurable space, and     ${\bf P}_{X_1, X_2, W}$ any joint  probability distribution of the triple $(X_1, X_2,W)$ on the product space
$(\mathbb{R}^{p_1} \times \mathbb{R}^{p_2}\times {\mathbb W}, 
  B(\mathbb{R}^{p_1}) \otimes B(\mathbb{R}^{p_2})\otimes B({\mathbb W}))$      with $(X_1, X_2)$-marginal ${\bf P}_{X_1, X_2}$ the Gaussian distribution ${\bf P}_{X_1, X_2}=G(0,Q_{(X_1,X_2)})$.  Define the RVs,    by, 
\begin{align}
  \begin{pmatrix}  
              Z_1\\Z_2   
            \end{pmatrix} \tri    \begin{pmatrix}
          X_1 \\X_{2}  
        \end{pmatrix}          
    -  {\bf E}\Big[          
          \begin{pmatrix}
            X_1 \\X_{2}  
          \end{pmatrix} \Big|
           W\Big], \; Z_i : \Omega \rightarrow {\mathbb R}^{p_i}.          
 \label{pr_g_7}
 \end{align}
(a)   There exists a Gaussian measure ${\bf P}_1=G(0,Q_{(X_1,X_2,W)})$ defined on the space $(\mathbb{R}^{p_1} \times \mathbb{R}^{p_2} \times  \mathbb{R}^{n}, B(\mathbb{R}^{p_1})\otimes B(\mathbb{R}^{p_2})\otimes B(\mathbb{R}^{n})), n \in \mathbb{Z}_+$ associated with the Gaussian RVs $(X_1, X_2)$, $W:\Omega \rightarrow {\mathbb R}^n,\; W \in G(0, Q_W),\;Q_W \succ0$  such that ${\bf P}_1|_{\mathbb{R}^{p_1} \times \mathbb{R}^{p_2}}=G(0,Q_{(X_1,X_2)})$. Moreover, a  realization of the RVs $(X_1, X_2, W)$ with induced measure ${\bf P}_1=G(0,Q_{(X_1,X_2,W)})$ is\footnote{{If $Q_W\succeq 0$,  then by the theory of Gaussian RVs, one needs to replace $Q_W^{-1}$ by pseudoinverse $Q_W^\dagger$, and has the option to  use the minimum realizations discussed in \cite{charalambous:schuppen:2019:arxiv}}.},
\begin{align}
   &     \begin{pmatrix}
          X_1 \\X_{2}  
        \end{pmatrix}          
    =  Q_{(X_1,X_2), W}Q_W^{-1} W + 
            \begin{pmatrix}  
              Z_1\\Z_2   
            \end{pmatrix} 
            ,   \label{ral_g1} \\
     & (Z_1,Z_2) \in G(0, Q_{(Z_1, Z_2)}), \; (Z_1, Z_2) \; 
          \mbox{indep. of} \ \ W,  \label{ral_g2}   \\
     &   Q_{(Z_1, Z_2)}= Q_{(X_1,X_2)}- Q_{(X_1,X_2), W} Q_W^{-1} Q_{(X_1,X_2), W}^T \succ 0,    \label{ral_g2_a}    \\   
    &    Q_{(X_1,X_2), W}
    =  {\bf E}\Big[          
          \begin{pmatrix}
            X_1 \\X_{2}  
          \end{pmatrix} 
           W^T\Big].    \label{ral_g2_a_n}    
\end{align}
(b)  Consider  $I(X_1,X_2;W)\in [0,\infty)$, for arbitrary RV $W$. Then the   inequalities hold. 
 \begin{align}
 &I(X_1, X_2;W)= H(X_1,X_2) - H(X_1, X_2|W)   \label{pr_g_1}\\
 &= H(X_1,X_2) - H(X_1 -{\bf E}[X_1| W], X_2-{\bf E}[X_2|W]  |W)    \label{pr_g_4}
 \end{align}
 \begin{align}
& \geq  H(X_1,X_2) - H(X_1 -{\bf E}[X_1|W] ,X_2-{\bf E}[X_2|W])  \label{pr_g_5}\\
& =  H(X_1,X_2) - H(Z_1 ,Z_2),   \hst \mbox{$Z_i$ defined by (\ref{pr_g_7})}   \label{pr_g_6}\\
& \geq H(X_1,X_2) - H(Z_1, Z_2),  \;  \mbox{if $(Z_1,Z_2) \in G(0, Q_{(Z_1,Z_2)})$} \label{pr_g_7_a}\\
& = \frac{1}{2} \ln \Big( \frac{\det(Q_{(X_1, X_2)})}{\det(Q_{(Z_1, Z_2) }}\Big)^+  \label{pr_g_7_aa_n} 
 \end{align} 
 where $Q_{(Z_1,Z_2)}\succ 0$.
 If,  \\
 (i) $W: \Omega \rightarrow {\mathbb W}={\mathbb R}^n$ $n\in {\mathbb Z}_+$,  Gaussian RV, and \\
 (ii)    $(Z_1, Z_2, W)$ are  jointly Gaussian RVs,\\
  then all   inequalities in   (\ref{pr_g_1})-(\ref{pr_g_7_a}) hold with equality, and   $(X_1,X_2, W)$ induces a family of joint  probability distributions     ${\bf P}_{X_1, X_2, W}$  with $(X_1, X_2)-$marginal ${\bf P}_{X_1, X_2}=G(0,Q_{(X_1,X_2)})$.  
 \end{theorem}
\begin{proof}
 (a) This is constructive, and follows from the realization of RVs  that induce Gaussian measures, as  presented in  \cite{charalambous:schuppen:2019:arxiv} and \cite{charalambous2020characterization}. (b) Inequality (\ref{pr_g_4}) is due to a property of conditional entropy, (\ref{pr_g_5}) is due to    conditioning reduces entropy, (\ref{pr_g_6}) is by   definition (\ref{pr_g_7}), and   (\ref{pr_g_7_a}),  is due to  maximum entropy principle. Then  (\ref{pr_g_7_aa_n}) follows by calculation of the entropies.  To show the last statement, use   (\ref{pr_g_7}),     $X_1 ={\bf E}[X_1|W]+Z_1,\; X_2 ={\bf E}[X_2|W]+Z_2$ and $ (X_1,X_2) \in G(0,Q_{(X_1,X_2)})$ (where $(Z_1,Z_2)$ are correlated);  if  (i) and (ii) hold,  then all inequalities hold with equalities, and the statements are easily verified.      
\end{proof} 
 
 \begin{remark} It is known from Gray and Wyner \cite{gray-wyner:1974} that  rate triples $(R_0, R_1, R_2) \in {\cal R}(\Delta_1, \Delta_2)$, with $R_0=0$ are also achievable; hence not surprizing is the fact that the parametrization in   Theorem~\ref{wyner_common_g},   of $I(X_1,X_2;W)$,  with respect to $W$  also includes $I(X_1,X_2;W)=0$, i.e., $W$  generates zero information, by replacing  $Q_W^{-1}$  by the pseudoinverse $Q_W^\dagger$,  or $W$ is  independent of   $(X_1,X_2)$. 
 \end{remark}

The proof of Theorem~\ref{wyner_common_g} does not pre-suppose that $W$ is a Gaussian RV, and the achievable lower bounds are parametrized by the joint distribution ${\bf P}_{X_1,X_2,W}$.  Hence, it is fundamentally different from Corollary 1 in  \cite{satpathy:cuff:2012}, that deals with minimizing $I(X_1,X_2;W)$ subject to $W$ that makes $X_1$ and $X_2$ conditionally independent, and  makes use  of  RDFs of Gaussian RVs with square-error distortion functions (strictly speaking one needs to prove Gaussian is optimal). 

Next we prove Theorem~\ref{th:crdf_g},   which gives  lower bounds 
on  $R_{X_i|W}(\Delta_i)$ for arbitrary $W : \Omega \rightarrow \mathbb{W}$,  $(X_1 , X_2 ) \in G(0,Q_{(X_1,X_2)})$, and    square-error distortions,
which are achievable if,   ${\bf P}_{\widehat{X}_i,X_i,W}$ is jointly Gaussian,   $W : \Omega \rightarrow \mathbb{R}^n, n \in {\mathbb Z}_+$, is Gaussian,  and a certain structural property of a realization holds.  

\begin{theorem} 
\label{th:crdf_g}  Consider the  conditional  RDFs $R_{X_i|W}(\Delta_i),i=1,2$,  for a triple of   RVs 
$X_i: \Omega \rightarrow \mathbb{R}^{p_i}, i=1,2$, $W : \Omega \rightarrow \mathbb{W}$, where  $W$  is a continuous, countable, or finite valued RV, with joint distribution ${\bf P}_{X_1, X_2,W}$ such that   the marginal distributions ${\bf P}_{X_1, X_2}$ and ${\bf P}_{X_i}$ for $i=1,2$ are Gaussian, i.e., $(X_1 , X_2 ) \in G(0,Q_{(X_1,X_2)})$,  $Q_{(X_1,X_2)}\succ 0$, 
and   square error distortion functions $D_{X_i} (x_i, \widehat{x}_i)= ||x_{i}-\widehat{x}_{i}||_{{\mathbb R}^{p_i}}^2, i=1,2$. Then the following hold.\\
(a) For  arbitrary RV, $W : \Omega \rightarrow \mathbb{W}$, $I(X_i; \widehat{X}_i|W)$ satisfies
\begin{align}
I(X_i; \widehat{X}_i|W)\geq I(X_i; {\bf E}[X_i|\widehat{X}_i,W]|W), \ \   i=1,2 \label{crdf_G1}
\end{align}
and the mean square error satisfies for $i=1,2$, 
\begin{align}
{\bf E}[D_{X_i}(X_{i},\widehat{X}_{i})] \geq {\bf E}[D_{X_i}(X_{i},{\bf E}[X_i|\widehat{X}_{i},W] )]. \label{crdf_G2}
\end{align}
If there exists a realization  $\widehat{X}_i$ of  the  test channel   distribution ${\bf P}_{\widehat{X}_i|X_i,W}$,  such that  the joint distribution  ${\bf P}_{\widehat{X}_i,X_i,W}$ satisfies,
\begin{align}
\widehat{X}_i^{\mathrm{cm}}\tri  {\bf E}[X_i| \widehat{X}_{i},W]={\bf E}[X_i|\widehat{X}_{i}]=\widehat{X}_i\text{-a.s.}, \ \ \ \ i=1,2. \label{crdf_G3}
\end{align}
then inequalities in (\ref{crdf_G1}), (\ref{crdf_G2}) hold with equality. 
 Moreover,  the following lower bounds hold.
\begin{align}
     & I(X_i;\widehat{X}_i|W) 
       \geq  \int_{{\mathbb W}} I(X_i;\widehat{X}_i^{\mathrm{cm}}|W=w){\bf P}_W    \label{cond_in_1_a}   \\
   &  \geq  \inf_{w \in {\mathbb W} }I(X_i;\widehat{X}_i^{\mathrm{cm}}|W=w)\geq 0 , \ \ i=1,2,  
             \label{cond_in_1}\\
  &     {\bf E}[   D_{X_i}(X_i,\widehat{X}_i)   ]   
   =  \int_{\mathbb W} 
          \left( 
            \int_{{{\mathbb R}^{p_1} \times {\mathbb R}^{p_2}}}  
                    D_{X_i}(x_i,\widehat{x}_i) 
                    {\bf P}_{X_i, \widehat{X}_i|W} 
          \right) {\bf P}_W     \nonumber\\      
 & =  \int_{{\mathbb W}} \Delta_i(w) {\bf P}_{W}, \hso \Delta_i(w) \tri {\bf E}[ D_{X_i}(X_i,\widehat{X}_i)\big|W=w]   \nonumber \\
 &  \geq   \int_{{\mathbb W}} \Delta_i^{\mathrm{cm}}(w) {\bf P}_{W}, \hso \Delta_i^{\mathrm{cm}}(w) \tri {\bf E}[ D_{X_i}(X_i,\widehat{X}_i^{\mathrm{cm}})\big|W=w]   \label{cond_in_2_b} \\
 & \geq 
          \inf_{w \in {\mathbb W} } \Delta_i^{\mathrm{cm}}(w) \geq 0,  \; i=1,2. \label{cond_in_2}       
\end{align}
Inequalities  (\ref{cond_in_1}),   (\ref{cond_in_2}) are achieved  if, {(a.i) (\ref{crdf_G3}) hold}, and  (a.ii) {the} mutual information  $I(X_i;\widehat{X}_i^{\mathrm{cm}}|W=w)$ and $\Delta_i^{\mathrm{cm}}(w)$ for $i=1,2$,  are independent of $w \in {\mathbb W}$.\\
(b) The  RDF $R_{X_i|W}(\Delta_i)\in [0,\infty), i=1,2$ satisfies,  
\begin{align}
       \hspace{-0.1cm}    R_{X_i|W}&(\Delta_i) 
     =\inf_{\{ \Delta_i(w)| w \in {\mathbb W}\}:    {\bf E}[ \Delta_i(W) ]
    \leq    \Delta_i}  \int R_{X_i|W=w}(\Delta_i(w)){\bf P}_W  \label{cond_in_3_new} \\
  &\geq   \inf_{\{ \Delta_i^{\mathrm{cm}}(w)| w \in {\mathbb W}\}:    {\bf E}[ \Delta_i^{\mathrm{cm}}(W) ]
    \leq    \Delta_i}  \int R_{X_i|W=w}(\Delta_i^{\mathrm{cm}}(w)){\bf P}_W  \label{cond_in_3_new_a} \\
    &\geq   \inf_{\{ \Delta_i^{\mathrm{cm}}(w)| w \in {\mathbb W}\}:    {\bf E}[ \Delta_i^{\mathrm{cm}}(W) ]
    \leq    \Delta_i}  R_{X_i|W=w}(\Delta_i^{\mathrm{cm}}(w)) \label{cond_in_3_new_b}
\end{align}
where $R_{X_i|W=w}(\Delta_i(w))$  is the RDF calculated for the distribution ${\bf P}_{X_i|W}$, for fixed $W=w$, and the infimum in (\ref{cond_in_3_new})  is over the  sets $\{\Delta_i(w)\big|w \in {\mathbb W}\}$ that satisfy the average constraint,  $\int_{{\mathbb W}} \Delta_i(w) {\bf P}_W\leq \Delta_i$ for  $i=1,2$. \\
 The lower bound (\ref{cond_in_3_new_b}), and the lower bounds of part (a)  i.e.,   (\ref{cond_in_1}),   (\ref{cond_in_2}),  are achieved if, \\
(b.i)  $W:\Omega \rightarrow {\mathbb R}^n$ is Gaussian,  and\\
(b.ii)  the joint distribution  ${\bf P}_{\widehat{X}_i,X_i,W}$ induced by $(\widehat{X}_i,X_i,W)$ is jointly Gaussian and  {(\ref{crdf_G3}) hold}, for $i=1,2$.\\
{The} RDFs $R_{X_i|W}(\Delta_i)\in [0,\infty), i=1, 2$ are characterized~by 
\bea
       R_{X_i|W}(\Delta_i) 
   =    \inf_{ \Delta_i^{\mathrm{cm}} \leq  \Delta_i   } R_{X_i|W=w}(\Delta_i^{\mathrm{cm}}) ,  \;  \Delta_i^{\mathrm{cm}}(w)=\Delta_i^{\mathrm{cm}}, \forall w\in {\mathbb W}.\label{cond_in_3}
\eea
 (c) {A realization that achieves the lower bounds of parts (a) and (b),  is the  Gaussian realization of $(X_i,W, \widehat{X}_i)$, parametrized by $(H_i, Q_{V_i}, Q_W)$,   given below.}
\begin{align*}
& \widehat{X}_i = H_i X_i + \left(I_{p_{i}}-H_i \right)Q_{X_i,W}Q_W^{-1}W +V_i, \ \ i=1,2,   \\
&V_i \in G(0, Q_{V_i}), \ \ V_i \ \ \mbox{independent of} \ \ (X_i,W) \\
&H_iQ_{X_i|W}=Q_{X_i|W} - Q_{E_i}\succeq 0, \; H_iQ_{X_i|W}= Q_{X_i|W}H_i^T, \\
&Q_{V_i}=H_i Q_{X_i|W}- H_i Q_{X_i|W} H_i^T\succeq 0, \hso    Q_{X_i|W}\succeq 0, 
\\
&E_i = X_i-\widehat{X}_i, \ \   E_i \in G(0,Q_{E_i}) , \hso Q_{E_i}\succeq 0,
\end{align*}
{i.e., satisfies satisfies the structural property (\ref{crdf_G3}),   for $i=1,2$. The  characterization of the  RDF $R_{X_i|W}(\Delta_i), i=1, 2$  is} 
\begin{align}
R_{X_i|W}(\Delta_i) & {=}  \inf_{{\bf P
}_{\widehat{X}_i|X_i,W}: {\bf E}\big[||X_i-\widehat{X}_i||_{{\mathbb R}^{p_i}}^2\big]\leq \Delta_i} I(X_i;\widehat{X}_i|W)\in [0,\infty)\nonumber \\
&=
\inf_{\trace(Q_{E_i}) \leq \Delta_i} \frac{1}{2} \ln \Big(  \frac{\det(Q_{X_i|W})}{\det(Q_{E_i})} \Big)^+, \label{conRDF_g}  \\
&\hspace{-1.4cm}\mbox{such that} \ \ Q_{\widehat{X}_i|W}=Q_{X_i|W}- Q_{E_i}\succeq 0,  \ \ Q_{E_i}\succ 0,
\end{align}
for $i=1,2$, where the test channel distribution ${\bf P}_{\widehat{X}_i|X_i,W}$ or the joint distribution  ${\bf P}_{\widehat{X}_i, X_i,W}$ is induced by the above realization. 
\end{theorem}
\begin{proof}  
 (a)  {Consider any  distribution ${\bf P}_{X_i,\widehat{X}_i, W}$, with ${\bf P}_{X_i}$ the Gaussian distribution of $X_i$. Inequalities \eqref{crdf_G1} and \eqref{crdf_G2}, and the last statement of part (a)  are  shown in  \cite{gkagkos2020structural,gkangos-charalambous:2021} (for general RVs $X_i$).}   The lower bounds  (\ref{cond_in_1_a}) and (\ref{cond_in_2_b}) follow directly  from \eqref{crdf_G1} and \eqref{crdf_G2}. {Inequalities  (\ref{cond_in_1}) and \eqref{cond_in_2} hold, since $I(X_i;\widehat{X}_i^{\mathrm{cm}}|W=w) \geq 0$ and  $\Delta_i(w) \geq 0$,  for all $w \in {\mathbb W}$. }
 Clearly, if (a.i) and (a.ii) hold, then the lower bound (\ref{cond_in_1}) and (\ref{cond_in_2}) are achieved. 
(b) Consider the RDF  $R_{X_i|W}(\Delta_i)$, for an arbitrary joint distribution ${\bf P}_{X_i,W}$ with ${\bf P}_{X_i}$ the Gaussian distribution of $X_i$. It  is well-known \cite{gray-wyner:1974}  that $R_{X_i|W}(\Delta_i)$ for $i=1,2$, are convex non-increasing functions in $\Delta_i\in (0, \infty]$,  and the average distortion constraint occurs on the boundary, for $\Delta_i \in (0, \Delta_{i,\max}]$, for some $\Delta_{i,\max}\in (0,\infty]$.  Identity   (\ref{cond_in_3_new}) is known from  \cite{gray-wyner:1974}. 
By part (a), then  $I(X_i;\widehat{X}_i|W=w) \geq I(X_i;\widehat{X}_i^{\mathrm{cm}}|W=w)$  and  $\Delta_i(w) \geq  \Delta_i^{\mathrm{cm}}(w)$ hold for all $ w\in {\mathbb W}$. By the convex non-increasing property of the RDFs,    $R_{X_i|W=w}(\Delta_i(w))\geq R_{X_i|W=w}(\Delta_i^{\mathrm{cm}}(w))$ for all $w\in {\mathbb W}$ and in addition $\Delta_i \geq {\bf E}[\Delta_i(w)] \geq {\bf E} [ \Delta_i^{\mathrm{cm}}(w)]$. Using these facts, inequality (\ref{cond_in_3_new_a}) is obtained, because the infimum is over a larger set. Next, inequality (\ref{cond_in_3_new_b}) follows because $R_{X_i|W=w}(\Delta_i^{\mathrm{cm}}(w)) \geq 0$, for all $w\in {\mathbb W}$.
 {Furthermore, if (b.i) and (b.ii) hold then by using the fact that, for a triple of   Gaussian RVs $(X, \widehat{X}_i, W)$,     conditional    mean-square errors,  and  conditional mutual informations do not depend on the realizations of the conditioning RVs,  then  $\Delta_i^{\mathrm{cm}}(w)=\Delta_i^{\mathrm{cm}}$, $R_{X_i|W=w}(\Delta_i^{\mathrm{cm}}(w))=R_{X_i|W=w}(\Delta_i^{\mathrm{cm}})$ for all $w\in {\mathbb W}$, i.e., they do not depend on the realizations $W=w$, and the mutual information $I(X_i;\widehat{X}_i^{\mathrm{cm}}|W=w)$ is also independent of $w\in {\mathbb W}$.}
{Hence,  (\ref{cond_in_3}) is  shown.}
   (c)   The reader may verify that the listed realization of the test channel of  $R_{X_i|W}(\Delta_i)$ ensures all  lower bounds of parts (a) and (b) are achievable. 
The listed realization  of the test of $R_{X_i|W}(\Delta_i)$ is shown, constructively,   for jointly Gaussian RVs $(X_i,W)$ in \cite{gkangos-charalambous:2021}.
 \end{proof}

 \begin{remark} Clearly, the statements of   Theorem~\ref{th:crdf_g} also   hold for  the  conditional joint RDF $R_{X_1,X_2|W}(\Delta_1, \Delta_2)$, i.e., $X_i$ and $\widehat{X}_i$ are replaced by the vectors $(X_1, X_2)$ and $(\widehat{X}_1, \widehat{X}_2)$.  
 \end{remark}

{Next, we apply Theorem~\ref{wyner_common_g} and     Theorem~\ref{th:crdf_g} to prove the first main theorem: }
${\cal R}(\Delta_1, \Delta_2)$,  
 as specified by (\ref{eqn_new_3}), for sources and distortions (\ref{dist_1})-(\ref{dist_3}), 
is achieved by the infimum
over ${\cal P}$ replaced by   the subset ${\cal P}^\mathrm{G}$, defined by  (\ref{eqn_new_5}). 

\begin{theorem}
\label{thm_par_g}
Consider a tuple of Gaussian RVs
$X_i: \Omega \rightarrow \mathbb{R}^{p_i}$,   $(X_1 , X_2 ) \in G(0,Q_{(X_1,X_2)})$ such that $Q_{(X_1,X_2)}\succ0$, with induced Gaussian measure ${\bf P}_0=G(0,Q_{(X_1,X_2)})$ on the space
$(\mathbb{R}^{p_1} \times \mathbb{R}^{p_2}, B(\mathbb{R}^{p_1})\otimes B(\mathbb{R}^{p_2}))$, 
  and  $D_{X_i} (x_i, \widehat{x}_i)= ||x_{i}-\widehat{x}_{i}||_{{\mathbb R}^{p_i}}^2, i=1,2$. Let  ${\bf P}_{X_1, X_2, W} \in {\cal P}^\mathrm{G}$,  be the family  of distributions  induced by the realization (\ref{ral_g1})-(\ref{ral_g2_a_n}) . 
\\
The Gray-Wyner achievable rate region  ${\cal R}(\Delta_1, \Delta_2)$
is determined by ${T^G}(\alpha_1, \alpha_2) = {T}(\alpha_1, \alpha_2)$,
\begin{align}
&{T^G}(\alpha_1, \alpha_2) = \inf_{ {\bf P}_{X_1, X_2, W} \in {\cal P}^\mathrm{G}  } \Big\{I(X_1, X_2; W)+\alpha_1 R_{X_1|W}(\Delta_1)  \nonumber\\
& \hspace*{3.2cm}+ \alpha_2 R_{X_2|W}(\Delta_2)\Big\}\in [0,\infty)  \label{g-w_reg_g}\\
&= \inf_{ (X_1, X_2, W) \in G(0, Q_{(X_1, X_2, W)}), \; \mbox{\small of (\ref{ral_g1})-(\ref{ral_g2_a})} } \Big\{I(X_1, X_2; W)\nonumber \\
&\hspace*{3.2cm}+\alpha_1 R_{X_1|W}(\Delta_1) + \alpha_2 R_{X_2|W}(\Delta_2)\Big\} \label{cig-w_reg_g}\\
&= \inf_{ Q_{X_1,X_2|W}, Q_{X_i|W}, i=1,2} \Big\{ \frac{1}{2} \ln \Big( \frac{\det(Q_{(X_1, X_2)})}{\det(Q_{X_1|X_2, W}) \det(Q_{X_2|W})}\Big)^+  \nonumber \\
&  \hspace*{3.2cm} +\alpha_1 R_{X_1|W}(\Delta_1) + \alpha_2 R_{X_2|W}(\Delta_2)\Big\}, \label{ccig-w_reg_g}\\
&Q_{X_1|X_2,W}=Q_{X_1|W} -Q_{X_1, X_2|W} Q_{X_2|W}^{-1} Q_{X_1, X_2|W}^T \succ 0  \label{cov_CI}
\end{align}
where $0\leq \alpha_i\leq 1,\; i=1,2,\;  \alpha_1+\alpha_2\geq 1$,  and  $R_{X_i|W}(\Delta_i), i=1,2$ are given in  Theorem~\ref{th:crdf_g}.(c). 
\end{theorem}
\begin{proof}
{By  Theorem~\ref{wyner_common_g} and     Theorem~\ref{th:crdf_g} the lower  bounds on   the quantities $I(X_1, X_2; W)$, $R_{X_i|W}(\Delta_i), i=1,2$, are simultaneously achieved by a  jointly Gaussian distributions ${\bf P}_{X_1,X_2,W} \in {\cal P}^G$,  induce by the realization (\ref{ral_g1})-(\ref{ral_g2_a_n})}.  Consequently,   (\ref{g-w_reg_g}) {follows  from the definition of $T(\alpha_1, \alpha_2)$ given by (\ref{eqn_new_3}),    and  Theorem~\ref{wyner_common_g},    Theorem~\ref{th:crdf_g}, due to $\alpha_i$ are nonnegative. Hence, the  infimum in $T(\alpha_1, \alpha_2)$,    is over   the parametrized  set  of  the jointly Gaussian RVs $(X_1, X_2, W) \in G(0, Q_{(X_1, X_2, W)})$ with joint distribution (\ref{jdf_g}), and (\ref{cig-w_reg_g}) follows.  From  the simultaneous achievability of (\ref{pr_g_7_aa_n}) and (\ref{conRDF_g}) then  (\ref{ccig-w_reg_g}) follows,  where    $R_{X_i|W}(\Delta_i)=\mbox{(\ref{conRDF_g})}$ depends only on $Q_{X_1|W}$, and the errors (see Theorem~\ref{th:crdf_g}.(c)).}  
 \end{proof}

Now, we prove the second main theorem of  the characterization of  ${\cal R}(\Delta_1, \Delta_2)$ for  (\ref{dist_1})-(\ref{dist_3}), using  (\ref{eq_32}), (\ref{eq_32_a}). 

\begin{theorem}
\label{them_g}
Consider the statement of Theorem~\ref{thm_par_g}. \\
For each ${\bf P}_{X_1, X_2, W} \in {\cal P}^\mathrm{G}$ and 
$\Delta_1 \geq 0, ~ \Delta_2 \geq 0$, define  
 \begin{align}
  {\cal R}^{{\bf P}_{X_1,X_2,W}}(\Delta_1, \Delta_2) = \Big\{&\big(R_0,R_1,R_2\big) \Big| \ \       R_0 \geq I(X_1, X_2; W), \nonumber \\ 
& R_1 \geq R_{X_1|W}(\Delta_1), \ \ R_2 \geq R_{X_2|W}(\Delta_2) \Big\} \nonumber  
\end{align} 
where  $R_{X_i|W}(\Delta_i), i=1,2$ are given in Theorem~\ref{th:crdf_g}.(c), and 
\begin{align}
&I(X_1,X_2;W) = \frac{1}{2} \ln \Big( \frac{\det(Q_{(X_1, X_2)})}{\det(Q_{X_1|X_2, W}) \det(Q_{X_2|W})}\Big)^+  \nonumber \\
&Q_{X_1|X_2, W}\succ 0, \hso Q_{X_2|W}\succ 0, \hso Q_{X_1|X_2,W}= \mbox{ (\ref{cov_CI})}. \nonumber
 \end{align} 
The achievable Gray-Wyner lossy rate region is  ${\cal R}(\Delta_1, \Delta_2)={\cal R}^{*,\mathrm{G}}(\Delta_1, \Delta_2) $, where,  
\begin{align}
{\cal R}^{*,\mathrm{G}}(\Delta_1, \Delta_2) &= \Big(\bigcup_{ {\bf P}_{X_1,X_2, W} \in {\cal P}^\mathrm{G}} {\cal R}^{{\bf P}_{X_1,X_2,W}}(\Delta_1, \Delta_2)\Big)^\mathrm{c} \nonumber \\
&=\Big(\bigcup_{ Q_{X_1,X_2|W}, Q_{X_i|W}, i=1,2} {\cal R}^{{\bf P}_{X_1,X_2,W}}(\Delta_1, \Delta_2)\Big)^\mathrm{c}. \nonumber
\end{align}
\end{theorem}
\begin{proof} 
This follows similarly to Theorem~\ref{thm_par_g}. 
\end{proof}  

A subset of ${\cal R}^{*,\mathrm{G}}(\Delta_1, \Delta_2) $, which is easier to compute, is obtained by restricting ${\cal P}^\mathrm{G}$, to  
 Gaussian distributions such that $W$ makes $X_1$ and $X_2$ conditional independent, defined by 
\begin{align}
{\cal P}^{\mathrm{CIG}} \tri & \Big\{ {\bf P}_{X_1, X_2, W} \in {\cal P}^\mathrm{G} \Big| \ \ {\bf P}_{X_1, X_2|W}={\bf P}_{X_1|W} {\bf P}_{X_2|W} \Big\}\subseteq {\cal P}^\mathrm{G} . \nonumber
\end{align}
Corollary~\ref{cor_sgy} is a special case of  Theorem~\ref{thm_par_g} and  Theorem~\ref{them_g},   
 by restricting ${\cal P}^{\mathrm{G}}$  to 
${\cal P}^{\mathrm{CIG}} \subseteq {\cal P}^{\mathrm{G}} $. This restriction over bounds $T^\mathrm{G}(\alpha_1,\alpha_2)$ of Theorem~\ref{thm_par_g}, and determines a  subset of the  rate region that intersects ${\cal R}(\Delta_1, \Delta_2)$. 
\par
\begin{corollary} 
\label{cor_sgy}
Consider the statement of Theorem~\ref{thm_par_g}.
There exists a Gaussian measure ${\bf P}_1=G(0,Q_{(X_1,X_2,W)})$ as defined in 
Theorem~\ref{wyner_common_g}.(a), which additionally satisfies  ${\bf P}_{X_1, X_2|W}={\bf P}_{X_1|W} {\bf P}_{X_2|W}$, and is  induced by realizations
\begin{align}
&X_i= Q_{X_i,W}Q_W^{-1} W + Z_i, \; \; Z_i \in G(0,Q_{Z_i}), \; \;  i=1,2,   \\
& (Z_1, Z_2, W) \hso  \mbox{mutually independent}, \\
&Q_{Z_i}= Q_{X_i} -Q_{X_i,W}Q_W^{-1} Q_{X_i,W}^T > 0, \ \ i=1,2.
\end{align}
An upper bound on {$T^{G}(\alpha_1,\alpha_2)$} of Theorem~\ref{thm_par_g}  
for the same  $(\alpha_1, \alpha_2)$  is obtain by replacing the set ${\cal P}^{\mathrm{G}}$ by ${\cal P}^{\mathrm{CIG}}$ in \eqref{g-w_reg_g}
\begin{align}
&{ T^G(\alpha_1, \alpha_2)} \leq  T^{\mathrm{CIG}}(\alpha_1, \alpha_2) 
= \inf_{ Q_{X_1|W}, Q_{X_2|W}} \Big\{ \alpha_1 R_{X_1|W}(\Delta_1)  \nonumber \\
   &\hspace{0.5cm} + \alpha_2 R_{X_2|W}(\Delta_2) + \frac{1}{2} \ln \Big( \frac{\det(Q_{(X_1, X_2)})}{\det(Q_{X_1|W}) \det(Q_{X_2|W})}\Big)^+  \Big\}, \label{ccig-w_reg}
\end{align}
where 
$0\leq \alpha_i\leq 1,\; i=1,2,\;  \alpha_1+\alpha_2\geq 1$,  and  $R_{X_i|W}(\Delta_i), i=1,2$ are given in  Theorem~\ref{th:crdf_g}.(c).\\
A subset  ${\cal R}^{*,\mathrm{CIG}}(\Delta_1, \Delta_2)\subseteq {\cal R}^{*,\mathrm{G}}(\Delta_1, \Delta_2)$ of Theorem~\ref{them_g} is 
\begin{align}
{\cal R}^{*,\mathrm{CIG}}(\Delta_1, \Delta_2)\tri \big(\bigcup_{ Q_{X_i|W}, i=1,2} {\cal R}^{{\bf P}_{X_1,X_2,W}}(\Delta_1, \Delta_2)\big)^\mathrm{c}. \label{overbound_1}
 \end{align}
\end{corollary}
\begin{proof}
 The first part  follows from   
the derivation of  Theorem~\ref{wyner_common_g}, by restricting  the joint probability {distributions of the triple $(X_1,X_2, W)$ to satisfy   conditional independence ${\bf P}_{X_1, X_2|W}={\bf P}_{X_1|W} {\bf P}_{X_2|W}$. }   Since  ${\cal P}^{\mathrm{CIG}} \subseteq {\cal P}^\mathrm{G}$, then the inequality  (\ref{ccig-w_reg}) holds, and also
 $Q_{X_1,X_2|W}=0$, in Theorem~\ref{thm_par_g}. Clearly, $T^{\mathrm{CIG}}(\alpha_1, \alpha_2)$ determines an over bound and a non-empty set on ${\cal R}^{*,\mathrm{G}}(\Delta_1, \Delta_2)$, since ${\cal P}^{\mathrm{CIG}} \subseteq {\cal P}^\mathrm{G}$.    
 Finally, we obtain (\ref{overbound_1}), as a special case of Theorem~\ref{them_g}.
\end{proof}

\section{Pangloss Plane of the Gray-Wyner Network}
{In this section we characterize the Pangloss Plane of the Gray-Wyner network, 
for arbitrary sources and distortions. Our contribution  is the presentation of a proof that  uses the Gray and Wyner characterization \cite{gray-wyner:1974} of $(R_0, R_1, R_2) \in {\cal R}(\Delta_1, \Delta_2)$, which  is much shorter than   \cite[eqns(21)]{viswanatha:akyol:rose:2014}. 
}

\begin{theorem} 
\label{thm_pangloss}
Consider an arbitrary tuple of sources and distortions. The set of rate triples $(R_0, R_1, R_2) \in {\cal R}(\Delta_1, \Delta_2)$ 
 which lie on the Pangloss plane,   $R_0+R_1+R_2= R_{X_1, X_2}(\Delta_1, \Delta_2)$,   are characterized by,\vspace{-0.15cm}
\begin{align}
&\sum_{i=1}^2 R_{X_i|W}(\Delta_i)+ I(X_1, X_2; W)=\sum_{i=0}^2 R_i=R_{X_1, X_2}(\Delta_1, \Delta_2) \label{pangloss_1}
\end{align}
over  a strictly positive surface of the distortion region,   denoted by\footnote{see Gray \cite{gray:1973} for definition.} ${\cal D}_{X_1, X_2|W}(\Delta_1, \Delta_2) \subseteq [0,\infty)\times [0,\infty)$, such that  the  joint distribution  ${\bf P}_{W,X_1,X_2,\widehat{X}_1,\widehat{X}_2} $ satisfies the  following conditions,\vspace{-0.15cm}
\begin{align}
 \;{\bf P}_{\widehat{X}_1,\widehat{X}_2|W}={\bf P}_{\widehat{X}_1|W} {\bf P}_{\widehat{X}_2|W}, \hso  {\bf P}_{X_1,X_2|\widehat{X}_1, \widehat{X}_2, W}={\bf P}_{X_1, X_2|\widehat{X}_1, \widehat{X}_2} \label{pangloss_2}
\end{align}
and the marginals  ${\bf P}_{X_1,X_2,\widehat{X}_1,\widehat{X}_2}$ and  ${\bf P}_{X_i,\widehat{X}_i,W}$ are generated by the test channels of $R_{X_1,X_2}(\Delta_1, \Delta_2)$ and $R_{X_i|W}(\Delta_i)$, respectively. 
\end{theorem}
\begin{proof}  Recall Gray's lower bounds \cite{gray:1973}, \vspace{-0.15cm}
\begin{align}
R_{X_1|W}(\Delta_1)+ R_{X_2|W}(\Delta_2)\geq R_{X_1, X_2|W}(\Delta_1, \Delta_2),  \label{in_1}  \\
R_{X_1, X_2|W}(\Delta_1, \Delta_2)\geq R_{X_1, X_2}(\Delta_1, \Delta_2)- I(X_1,X_2;W)\label{in_2}
\end{align}   
It is easy to show (see for example \cite[Theorem~1]{leiner:1977}) that   (\ref{in_1}) holds with equality if  the left hand side of \eqref{pangloss_2} holds and \eqref{in_2} holds with equality  if the right  hand side of \eqref{pangloss_2} holds for a certain distortion region.  
Take a triple $(R_0,R_1, R_2)\in {\cal R}(\Delta_1, \Delta_2)$ 
such that $\sum_{i=1}^3R_i=R_{X_1, X_2}(\Delta_1, \Delta_2)$. Then,\vspace{-0.17cm}
\begin{align}
\hspace{-0.3cm}\sum_{i=1}^3R_i = R_{X_1, X_2}(\Delta_1, \Delta_2) &\leq  R_{X_1, X_2|W}(\Delta_1, \Delta_2)+ I(X_1, X_2;W)   \label{chain_4}   \\
& \hspace{-0.8cm}\leq  R_{X_1|W}(\Delta_1)+ R_{X_2|W}(\Delta_2) + I(X_1, X_2;W)  \label{chain_5}
\end{align}
where (\ref{chain_4})   is due to    (\ref{in_2}) and holds with equality on   a strictly positive surface   if the second condition in (\ref{pangloss_2}) holds, and (\ref{chain_5}) is to due to  inequality (\ref{in_1}) and holds with equality  if  the first  condition in (\ref{pangloss_2}) holds. The reverse inequality to  (\ref{chain_5}) is 
obtained as follows. For any    $(R_0, R_1, R_2) \in {\cal R}(\Delta_1, \Delta_2)$, by (\ref{eq_32}), follows   $  R_0\geq I(X_1,X_2;W), R_i \geq R_{X_1|W}(\Delta_1),i=1,2$, and hence, \vspace{-0.3cm}
\begin{align}
 R_{X_1, X_2}(\Delta_1, \Delta_2)&= \sum_{i=0}^2 R_i \geq  I(X_1,X_2;W)+ R_{X_1|W}(\Delta_1)+ R_{X_2|W}(\Delta_2).\nonumber
\end{align} 
  Hence, if (\ref{pangloss_2}) holds the upper and lower bounds coincide, and (\ref{pangloss_1})  is obtained.
This  completes the proof. 
\end{proof}

\section{Concluding Remarks}\label{sec:concludingremarks}
Characterized in this paper,  is the  Gray and Wyner \cite{gray-wyner:1974} achievable lossy rate region ${\cal R}(\Delta_1, \Delta_2)$  of a tuple of jointly Gaussian RVs, $X_1 : \Omega \rar {\mathbb R}^{p_1}, X_2 : \Omega \rar {\mathbb R}^{p_2}$ with square-error fidelity at the two decoders. The achievable rate region is parametrized by the 3 conditional covariances, $Q_{X_1,X_2|W}, Q_{X_1|W}, Q_{X_2|W}$ of a triple of Gaussian RVs $(X_1,X_2, W)$, where $W: \Omega \rar {\mathbb R}^n$ is a Gaussian RV. However, an   over bound on ${\cal R}(\Delta_1, \Delta_2)$  is obtained by the simpler parametrization  with respect to  $Q_{X_1|W}, Q_{X_2|W}$, which specifies a subset of the rate region. {Versions of these results are found in   \cite{charalambous:schuppen:2019:arxiv}.  The  characterizations of this  paper settled a long term open problem, regarding the Gray and Wyner rate region. }

\begin{footnotesize}
\bibliographystyle{IEEEtran}
\bibliography{bibliography}
\bibliographystyle{plain}
\end{footnotesize}

\end{document}